\documentclass{elsarticle}

\newif\ifarxiv
\arxivtrue

\ifarxiv
\usepackage{fullpage}
\makeatletter
\def\ps@pprintTitle{%
   \let\@oddhead\@empty
   \let\@evenhead\@empty
   \def\@oddfoot{\reset@font\hfil\thepage\hfil}
   \let\@evenfoot\@oddfoot
}
\makeatother
\fi

\usepackage[usenames,dvipsnames]{xcolor}
\usepackage[colorlinks=true,urlcolor=Blue,citecolor=Green,linkcolor=BrickRed]
{hyperref}
\makeatletter
\AtBeginDocument{\def\@citecolor{Green}}
\AtBeginDocument{\def\@linkcolor{BrickRed}}
\makeatother

\usepackage{epsfig}
\usepackage{epstopdf}
\usepackage{graphicx}
\DeclareGraphicsExtensions{.pdf,.png,.jpg,.eps}

\usepackage{latexsym}
\usepackage{amsfonts}
\usepackage{amssymb}
\usepackage{amsmath}
\usepackage{url}

\usepackage{bm}
\usepackage{algorithm} 
\usepackage{algpseudocode}
\usepackage{multirow}
\usepackage{url}
\usepackage{pifont}
\usepackage{todonotes}

\newcommand{\MS}[1]{\ensuremath{\Delta}}

\newtheorem{theorem}{Theorem}[section]
\newtheorem{lemma}[theorem]{Lemma}

\newtheorem{proposition}[theorem]{Proposition}

\newenvironment{proof}{\noindent{\bf Proof.\/}}{}

\newcommand{\RMQ}{{\textsf{RMQ}}}
\newcommand{\LCA}{{\textsf{LCA}}}

\journal{Journal of \LaTeX\ Templates}
\begin{document}

\begin{frontmatter}

\title{Compressed Range Minimum Queries\tnoteref{mytitlenote}}
\tnotetext[mytitlenote]{Preliminary version of this paper appeared in SPIRE 2018. The work is supported in part by Israel Science Foundation grant 592/17.}

\author[Pawel]{Pawe\l{}' Gawrychowski}
\address[Pawel]{University of Wroc\l{}aw}
\ead{gawry@cs.uni.wroc.pl}

\author[Seungbum]{Seungbum Jo}
\address[Seungbum]{University of Haifa}
\ead{seungbum.jo@uni-siegen.de}

\author[Shay]{Shay Mozes}
\address[Shay]{Interdisciplinary Center Herzliya}
\ead{smozes@idc.ac.il}

\author[Seungbum]{Oren Weimann}
\ead{oren@cs.haifa.ac.il}

\begin{abstract}
Given a string $S$ of $n$ integers in $[0,\sigma)$, a range minimum query $\RMQ(i, j)$ asks for the index of the smallest integer in $S[i \dots j]$. It is well known that the problem can be solved with a succinct data structure of size $2n + o(n)$ and constant query-time. 
In this paper we show how to preprocess $S$ into a {\em compressed representation} that allows fast range minimum queries. This allows for {\em sublinear} size data structures with logarithmic query time.
The most natural approach is to use string compression and construct a data structure for answering range minimum queries directly on the compressed string. We investigate this approach in the context of grammar compression. We then consider an alternative approach. 
Instead of compressing $S$ using string compression, we compress the Cartesian tree of $S$ using tree compression. 
We show that this approach can be exponentially better than the former, is never worse by more than an $O(\sigma)$ factor (i.e. for constant alphabets it is never asymptotically worse), and can in fact be worse by an $\Omega(\sigma)$ factor.  
\end{abstract}


\end{frontmatter}

\ifarxiv
\else
\begin{center}
\em{Dedicated to our friend Danny Breslauer}
\end{center}
\fi

\section{Introduction}

Given a string $S$ of $n$ integers in $[0,\sigma)$, a range minimum query $\RMQ(i, j)$ returns the index of the smallest integer in $S[i \dots j]$. A range minimum data structure consists of a preprocessing algorithm and a query algorithm. The preprocessing algorithm takes as input the string $S$, and constructs the data structure, whereas the query algorithm takes as input the indices $i,j$ and, by accessing the data structure, returns $\RMQ(i,j)$. The range minimum problem is a fundamental data structure problem that has been extensively studied, both in theory and in practice (see e.g.~\cite{fh-11} and references therein).

Range minimum data structures fall into two categories. {\em Systematic} data structures store the input string $S$ in plain form, whereas {\em non-systematic} data structures do not. A significant amount of attention has been devoted to devising RMQ data structures that answer queries in constant time and require as little space as possible. There are succinct systematic structures that answer queries in constant time and require fewer than $2n$ bits in addition to the $n\log \sigma$ bits required to represent $S$~\cite{fh-11}. Similarly, there are succinct non-systematic structures that answer queries in constant time, and require $2n + o(n)$ bits~\cite{fh-11, DavoodiRS12}.

The {\em Cartesian tree} $\mathcal{C}$ of $S$ is a rooted ordered binary tree with $n$ nodes. It is defined recursively. The Cartesian tree of an empty string is an empty tree. Let $i$ be the index of the smallest element of $S$ (if the smallest element appears multiple times in $S$, let $i$ be the first such appearance). The Cartesian tree of $S$ is composed of a root node whose left subtree is the Cartesian tree of $S[1,i-1]$, and whose right subtree is the Cartesian tree of $S[i+1,n]$. See Figure~\ref{figure:cartesian}.
By definition, the character $S[i]$ corresponds 
to the $i$'th node in an inorder traversal of $\mathcal{C}$ (we will refer to this node as node $i$). 
Furthermore, for any nodes $i$ and $j$ in $\mathcal{C}$, their lowest common ancestor $\LCA(i, j)$ in $\mathcal{C}$ corresponds to $\RMQ(i, j)$ in $S$. It follows that the Cartesian tree of $S$ completely characterizes $S$ in terms of range minimum queries. Indeed, two strings return the same answers for all possible range minimum queries if and only if their Cartesian trees are identical. This well known property has been used by many RMQ data structures including the succinct structures mentioned above. Since there are $2^{2n -O(\log n)}$ distinct rooted binary trees with $n$ nodes, there is an information theoretic lower bound of $2n -O(\log n)$ bits for RMQ data structures. In this sense, the above mentioned $2n + o(n)$ bits data structures~\cite{fh-11,DavoodiRS12} are nearly optimal.

\subsection{Our results and techniques}

In this work we present RMQ data structures in the word-RAM model (without using any bit tricks). The size (in words) of our data structures can be {\em sublinear} in the size of the input string and the query time is $O(\log n)$. This is achieved by using compression techniques, and developing data structures that can answer RMQ/LCA queries directly on the compressed objects. 
Since we aim for sublinear size data structures, we focus on non-systematic data structures.
We consider two different approaches to achieve this goal. The first approach is to use string compression to compress $S$, and devise an RMQ data structure on the compressed representation. This approach has also been suggested in~\cite[Section 7.1]{A} in the context of compressed suffix arrays and the LCP array. See also ~\cite[Theorem 2]{DavoodiRS12}, ~\cite[Theorem 4.1]{fh-11}, and~\cite{B} for steps in this direction.
The other approach is to use tree compression to compress the Cartesian tree $\mathcal C$, and devise an LCA data structure on the compressed representation. 
To the best of our knowledge, this is the first time such approach has been suggested.
Note that these two approaches are not equivalent. 
For example, consider a sorted sequence of an arbitrary subset of $n$ different integers from $[1,2n]$. As a string this sorted sequence is not compressible, but its Cartesian tree is an (unlabeled) path, which is highly compressible. In a nutshell, we show that the tree compression approach can exponentially outperform the string compression approach. Furthermore, it is never worse than the string compression approach by more than an $O(\sigma)$ factor, and this $O(\sigma)$ factor is unavoidable. We next elaborate on these two approaches.

\paragraph{\bf Using string compression}
In Section~\ref{sec:slp}, we show how to answer range minimum queries on a {\em grammar compression} of the input string $S$. A grammar compression is a context-free grammar that generates only $S$. The grammar is represented as a {\em straight line program} (SLP) $\mathcal S$. I.e., the 
right-hand side of each rule in $\mathcal{S}$ 
either consists of the concatenations of two non-terminals or of a single terminal symbol.
The size $|\mathcal S|$ of the SLP $\mathcal{S}$ is defined as 
the number of rules in $\mathcal{S}$. Ideally, $|\mathcal S| \ll |S|$. 
Computing the smallest possible SLP is NP-hard~\cite{Charikar-slp}, but there are many theoretically and practically efficient compression schemes for constructing $\mathcal{S}$~\cite{Charikar-slp, DBLP:journals/jda/GotoBIT13, DBLP:conf/esa/TakabatakeIS17} that reasonably approximate the optimal SLP. In particular, Rytter~\cite{Rytter03} showed an SLP $\mathcal{S}$ of depth $\log n$ (the depth of an SLP is the depth of its parse tree) whose size is larger than the optimal SLP by at most a multiplicative $\log n$ factor. Very recently, Ganardi et al.~\cite{GanardiJL19} showed that any SLP can be turned (with no asymptotic overhead) into an equivalent SLP that is of depth $\log n$.   

In~\cite{A}, it was shown how to support range minimum queries on $S$ with a data structure of size $O(|\mathcal S|)$ in time proportional to the depth of the SLP $\mathcal S$. Bille et al.~\cite{random-access} designed a data structure of size $O(|\mathcal S|)$  that supports random-access to $S$ (i.e. retrieve the $i$'th symbol in $S$) in $O(\log n)$ time (i.e. regardless of the depth of the SLP $\mathcal S$). We show how to simply augment their data structure within the same $O(|\mathcal S|)$ size bound to answer range minimum queries in $O(\log n)$ time (i.e. how to avoid the logarithmic overhead incurred by using the solution of~\cite{A} on Rytter's SLP).

\begin{theorem}
\label{thm:slprmq}
Given a string $S$ of length $n$ and an SLP-grammar compression $\mathcal{S}$ of $S$, there is a data structure of size $O(|\mathcal{S}|)$ that answers range minimum queries on $S$ in $O(\log n)$ time.
\end{theorem}

\paragraph{\bf Using tree compression}
In Section~\ref{sec:CT}, we give a data structure for answering LCA queries on a compressed representation of the Cartesian tree $\mathcal{C}$. By the discussion above, this is equivalent to answering range minimum queries on $S$. 
There are various ways to compress trees. In this work we use DAG compression of the top-tree of the Cartesian tree $\mathcal C$ of $S$. These concepts will be explained in the next paragraph. It is likely that other tree compression techniques (see, e.g.,~\cite{DBLP:conf/stacs/JezL14,DBLP:conf/dlt/Lohrey15,DBLP:journals/algorithmica/GanardiHLN18}) can also yield interesting results. We leave this as future work.

A {\em top-tree}~\cite{TopTrees} of a tree $T$ is a hierarchical decomposition of the edges of $T$ into clusters. 
A cluster is a connected subgraph of T that has at most two boundary nodes (nodes with neighbors outside the cluster); the root of the cluster (called the top boundary node), and a leaf of the cluster (called a bottom boundary node).
The intersection of any two clusters is either empty, or consists of exactly one (boundary) node.
Such a decomposition can be described by a rooted ordered binary tree $\mathcal T$, called a top-tree, whose leaves correspond to clusters with individual edges of $T$, and whose root corresponds to the entire tree $T$. The cluster corresponding to a non-leaf node of $\mathcal T$ is obtained from the clusters of its two children by either identifying their top boundary nodes (horizontal merge) or by identifying the top boundary node of the left child with the bottom boundary node of the right child (vertical merge).  See Figure~\ref{figure:cartesian}. 

A {\em DAG compression}~\cite{dag-compress} of a tree $T$ is a representation of $T$ by a DAG whose nodes correspond to nodes of $T$. All nodes of $T$ with the same subtree are represented by the same node of the DAG. Thus, the DAG compression of a top-tree has two sinks (due to $T$ being a binary unlabeled tree), corresponding to the two types of leaf nodes of $T$ (a single edge cluster, either left or right), and a single source, corresponding to the root of $T$. If $u$ is the parent of $\ell$ and $r$ in $T$, then the node in the DAG representing the subtree of $T$ rooted at $u$ has edges leading to the two nodes of the DAG representing the subtree of $T$ rooted at $\ell$ and the subtree of $T$ rooted at $r$. Thus, repeating rooted subtrees in $T$ are represented only once in the DAG. See Figure~\ref{figure:cartesian}. 

A {\em top-tree compression}~\cite{top-tree}  of a tree $T$ is a DAG compression of $T$'s top-tree $\mathcal T$.
Bille et al.~\cite{top-tree} showed how to construct a data structure whose size is linear in the size of the DAG of $\mathcal T$ and supports navigational queries on $T$ in time linear in the depth of $\mathcal T$. In particular, given the preorder numbers of two vertices $u,v$ in $T$, their data structure can return the preorder number of $\LCA(u,v)$ in $T$. We show that their data structure can be easily adjusted to work with inorder numbers instead of preorder, so that, given the inorder numbers $i,j$ of two vertices in $T$ one can return the inorder number of $\LCA(i,j)$ in $T$. This is precisely $\RMQ(i,j)$ when $T$ is taken to be the Cartesian tree $\mathcal{C}$ of $S$. 

\begin{theorem}
\label{thm:rmq_ct}
Given a string $S$ of length $n$ and a top-tree compression $\mathcal{T}$ of the Cartesian tree $\mathcal{C}$, we can construct a data structure of size $O(|\mathcal{T}|)$ that answers range minimum queries on $S$ in $O(\textnormal{\texttt{depth}}(\mathcal T))$ time.
\end{theorem}

By combining Theorem~\ref{thm:rmq_ct} with the greedy construction of $\mathcal{T}$ given in~\cite{top-tree} (in which $\texttt{depth}(\mathcal T) = O(\log n)$), we can obtain an $O(|\mathcal{T}|)$ space data structure that answers RMQ in $O(\log{n})$ time.

We already mentioned that, on some RMQ instances, top-tree compression can be much better than any string compression technique.
As an example, consider the string $S=1,2,3,\cdots, n$. Its Cartesian tree is a single (rightmost, and unlabeled) path, which  compresses using top-tree compression into size $|\mathcal{T}| = O(\log{n})$. On the other hand, since $\sigma = n$, $S$ is uncompressible with an SLP. By Theorem~\ref{thm:rmq_ct}, this shows that the tree compression approach to the RMQ problem can be exponentially better than the string compression approach. In fact, for any string over an alphabet of size $\sigma = \Omega(n)$, any SLP must have $|\mathcal{S}| = \Omega(n)$ while for top-trees $|\mathcal{T}| = O(n/\log{n})$~\cite{top-tree}.
In Section~\ref{sec:compare upper} we show that, for small alphabets, $\mathcal{T}$ cannot be much larger nor much deeper than $\mathcal{S}$ for any SLP $\mathcal{S}$. 

\begin{theorem}\label{thm:tg}
Given a string $S$ of length $n$ over an alphabet of size $\sigma$, for any SLP-grammar compression $\mathcal{S}$ of $S$ there is a top-tree compression $\mathcal{T}$ of the Cartesian tree $\mathcal{C}$ with size $O(|\mathcal{S}|\cdot \sigma)$ and depth $O(\texttt{depth}(\mathcal S) \cdot \log \sigma)$.
\end{theorem}

Observe that in the above theorem, in order to obtain small depth of $\mathcal{T}$, one could use, e.g, the SLP of Rytter~\cite{Rytter03} or that of Ganardi et al.~\cite{GanardiJL19}. Another way of achieving small depth is to ignore the SLP and use the top-tree compression of Bille et al.~\cite{top-tree} as $\mathcal{T}$. This guarantees $\mathcal{T}$ has size $O(n/\log n)$ words and depth $O(\log n)$.

%
%
%

Finally, observe that $\mathcal{T}$ can be larger than $\mathcal{S}$ by an  $O(\sigma)$ multiplicative factor which can be large for large alphabets. It is tempting to try and improve this. However, in Section~\ref{sec:compare lower} we prove a tight lower bound, showing that this factor is unavoidable.

\begin{theorem}\label{thm:lower bound}
For every sufficiently large $\sigma$ and $s=\Omega(\sigma^{2})$, there exists a string $S$ of integers in $[0,\sigma)$ that
can be described with an SLP $\mathcal{S}$ of size $s$, such that any top-tree compression $\mathcal{T}$ of the Cartesian
tree $\mathcal{C}$ of $S$ is of size $\Omega(s\cdot \sigma)$.
\end{theorem}

\section{RMQ on Compressed Representations}

\subsection{Compressing the string}\label{sec:slp}
Given an SLP compression $\mathcal{S}$ of $S$, Bille et al.~\cite{random-access} presented a data structure of size $O(|\mathcal{S}|)$ that can report any $S[i]$ in $O(\log n)$ time. We now prove Theorem~\ref{thm:slprmq} using a rather straightforward extension of this data structure to support range minimum queries.

The key technique used in~\cite{random-access} is an efficient representation of the {\em heavy path decomposition} of the SLP's parse tree. For each node $v$ in the parse tree, we select the child of $v$ that derives the longer string to be a {\em heavy} node. The other child is {\em light}. Ties can be broken arbitrarily.  Heavy edges are edges going into a heavy node and light edges are edges going into a light node. The heavy edges decompose the parse tree into \textit{heavy paths}. The number of light edges on any path from a node $v$ to a leaf is $O(\log |v|)$ where $|v|$ denotes the length of the string derived from $v$. A traversal of the parse tree from its root to the $i$'th leaf $S[i]$ enters and exists at most $\log n$ heavy paths. Bille~et~al. show how to simulate this traversal in $O(\log n)$ time on a representation of the heavy path decomposition that uses only $O(|\mathcal{S}|)$ space. In other words, their structure finds the entry and exit vertices of all heavy paths encountered during the root-to-leaf traversal in total $O(\log n)$ time. We elaborate on this now. 

Note that we cannot afford to store the entire parse tree as its size is $n$ which can be exponentially larger than $|\mathcal{S}|$. Instead, Bille~et~al. use the following $O(|\mathcal{S}|)$-space representation $H$ of the heavy paths: $H$ is a forest of trees whose roots correspond to terminals of $\mathcal{S}$ (integers in $[0,\sigma)$) and whose non-root  nodes correspond to nonterminals of $\mathcal{S}$. A node $u$ is the parent of $v$ in $H$ iff $u$ is the heavy child of $v$ in the parse tree (observe that wherever the nonterminal $u$ appears in the parse tree it always has the same heavy child $v$). We assign the edge of $H$ from $v$ to its parent $u$ with a left weight $\ell(v,u)$ and right weight $r(v,u)$ defined as follows. If $u$ is the left child of $v$ in $H$ then the left weight is $0$ and the right weight is the subtree size of the right child of $v$ in the parse tree. Otherwise, the right weight is $0$ and the left weight is the subtree size of the left child of $v$ in the parse tree.  

Using $H$, we can then simulate a root-to-leaf traversal of the parse tree. Suppose we have reached a vertex $u$ on some heavy path $P$. Finding out where we need to exit $P$ (and enter another heavy path) easily translates to a \emph{weighted level ancestor} query (using the $\ell(v,u)$ or $r(v,u)$ weights) from $u$ on $H$ (see~\cite{random-access} for details). Given a positive number $x$, such a query returns the rootmost ancestor of $u$ in $H$ whose distance from the root is at least $x$. Bille et al. showed how to answer all such queries (i.e. how to find the entry point and exit point on all heavy paths visited during a root-to-leaf traversal of the parse tree) in total $O(\log n)$ time. 

Extending their structure to support range minimum queries is quite simple. We perform a random access to the $i$'th and the $j$'th leaves in the parse tree. This identifies the entry and exit points of all traversed heavy paths, and, in particular, the unique heavy path $P'$ containing the lowest common ancestor of $i$ and $j$. Then, we wish to find the minimum leaf value in all the subtrees hanging to the right (resp. left) of the path starting from $P'$ and going down to $i$ (resp. $j$). To achieve this for $i$ (the case of $j$ is symmetric), in addition to the subtree sizes ($r(v,u)$) we also store the minimum leaf value ($r'(v,u)$) in these subtrees. This way, the problem now boils down to performing $O(\log n)$ {\em bottleneck edge queries} on $H$. Given a forest $H$ with edge weights (the $r'(v,u)$ weights), a bottleneck edge query between two vertices $x,y$ (the entry and exit points) returns the minimum edge weight on the unique $x$-to-$y$ path in $H$. Demaine et al.~\cite{DemaineLandauWeimann} showed that, after sorting the edge weights in $H$, one can construct in $O(|H|)=O(|\mathcal{S}|)$ time and space a data structure that answers bottleneck edge queries in constant time. This concludes the proof of Theorem~\ref{thm:slprmq}.

\subsection{Compressing the Cartesian tree}\label{sec:CT}

\begin{figure}[t!]
\centering
\ifarxiv \includegraphics[scale=0.45]{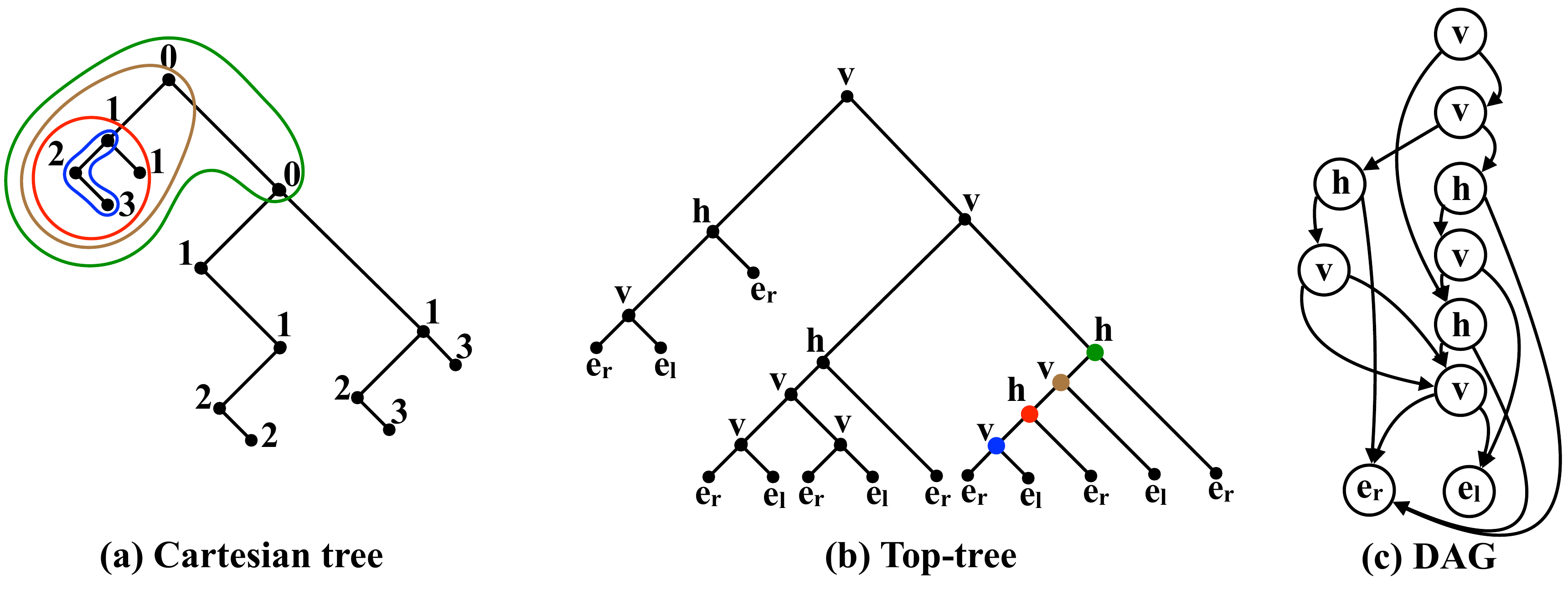}
\else \includegraphics[scale=0.35]{image}
\fi

\caption{The string $S=``2 3 1 1 0 1 2 2 1 0 2 3 1 3"$ and its corresponding (a) Cartesian tree , (b) top-tree, and (c) DAG representation of the top-tree. 
In (a), each node is labeled by its corresponding character in $S$ (these labels are for illustration only, the top-tree construction treats the Cartesian tree as an unlabeled tree). 
In (b) and (c), each node is labeled by $e_l$ or $e_r$ (atomic edge clusters), $v$ (a vertical merge), or $h$ (a horizontal merge). Four clusters are marked with matching colors in (a) and in (b).   
\label{figure:cartesian}}
\end{figure}

We next prove Theorem~\ref{thm:rmq_ct}, i.e. how to support range minimum queries on $S$ using a compressed representation of the {\em Cartesian tree}~\cite{Cartesian-tree}. Recall that the Cartesian tree $\mathcal{C}$ of $S$ is defined as follows: If the smallest character in $S$ is $S[i]$ (in case of a tie we choose a leftmost position) then the root of $\mathcal{C}$ corresponds to $S[i]$, its left child is the Cartesian tree of $S[1,i-1]$ and its right child is the Cartesian tree of $S[i+1,n]$.
By definition, the $i$'th character in $S$ corresponds to the node in $\mathcal{C}$ with
inorder number $i$ (we will refer to this node as node $i$). 
Observe that for any nodes $i$ and $j$ in $\mathcal{C}$, the lowest common ancestor $\LCA(i, j)$ of these nodes in $\mathcal{C}$ corresponds to $\RMQ(i, j)$ in $S$.  
This implies that without storing $S$ explicitly, one can answer range minimum queries on $S$ by answering LCA queries on $\mathcal{C}$.
In this section, we show how to support LCA queries on $\mathcal{C}$ on a {\em top-tree} compression~\cite{top-tree} $\mathcal{T}$ of $\mathcal{C}$. The query time is $O(\texttt{depth}(\mathcal T))$ which can be made $O(\log{n})$ using the (greedy) construction of Bille et al.~\cite{top-tree} that gives $\texttt{depth}(\mathcal T)=O(\log{n})$. We first briefly restate the construction of Bille et al., and then extend it to support LCA queries. 

The {\em top-tree} of a tree $T$ (in our case $T$ will be the Cartesian tree $\mathcal{C}$) is a hierarchical decomposition of $T$ into
{\em clusters}. 
Let $v$ be a node in $T$ with children $v_1,v_2$.\footnote{Bille et al. considered trees with arbitrary degree, but since our tree $T$ is a Cartesian tree we can focus on binary trees.} Define $T(v)$ to be the subtree of $T$ rooted at $v$.
Define $F(v)$ to be the forest $T(v)$ without $v$.
A {\em cluster} with {\em top boundary node} $v$ can be either (1) $T(v)$, (2) $\{v\} \cup T(v_1)$, or (3) $\{v\} \cup T(v_2)$. For any node $u \neq v$ in a cluster with top boundary node $v$, deleting from the cluster all descendants of $u$ (not including $u$ itself) results in a {\em cluster} with {\em top boundary node} $v$ and {\em bottom boundary node} $u$. The top-tree is a binary tree defined as follows (see Figure~\ref{figure:cartesian}):

\begin{itemize}
\item The root of the top-tree is the cluster $T$ itself. 

\item The leaves of the top-tree are (atomic) clusters corresponding to the edges of $T$. An edge $(v,parent(v))$ of $T$ is a cluster where $parent(v)$ is the top
boundary node. If $v$ is a leaf then there is no bottom boundary node,
otherwise $v$ is a bottom boundary node. If $v$ is the right child of $parent(v)$ then we label the $(v,parent(v))$ cluster as $e_r$ and otherwise as $e_\ell$.

\item Each internal node of the top-tree is a {\em merged} cluster of its two children. Two edge disjoint
clusters $A$ and $B$ whose nodes overlap on a single boundary node can
be merged if their union $A \cup B$ is also a cluster (i.e. contains at most two boundary nodes). If $A$ and $B$ share their top boundary node then the merge is called {\em horizontal}. If the top boundary node of $A$ is the bottom boundary node of $B$ then the merge is called {\em vertical} and in the top-tree $A$ is the left child and $B$ is the right child. 
\end{itemize}

Bille et al.~\cite{top-tree} proposed a greedy algorithm for constructing the top-tree: Start with $n-1$  clusters, one for each edge of $T$, and at each step merge all possible clusters. More precisely, at each step, first do all possible horizontal merges and then do all possible vertical merges. 
After constructing the top-tree, the actual compression $\mathcal{T}$ is obtained by representing the top-tree as a directed acyclic graph (DAG) using the algorithm of~\cite{dag-compress}. Namely, all nodes in the top-tree that have a child with subtree $X$ will point to the same subtree $X$ (see Figure~\ref{figure:cartesian}). 
Bille et al.~\cite{top-tree} showed that using the above greedy algorithm, one can construct
$\mathcal{T}$ of size $|\mathcal{T}|$ that can be as small as $\log n$ (when the input tree $T$ is highly repetitive) and in the worst-case is at most $O(n/\log_\sigma^{0.19}{n})$. 
Dudek and Gawrychowski~\cite{toptree-worst} have  recently improved the worst-case bound to $O(n/\log_{\sigma}{n})$ by merging in the $i$'th step only clusters whose size is at most $\alpha^i$ for some constant $\alpha$. Using either one of these merging algorithms to obtain the top-tree and its DAG representation $\mathcal{T}$, a data structure of size $O(\mathcal{|T|})$ can then be constructed to support various queries on $T$. In particular, given nodes $i$ and $j$ in $T$ (specified by their position in a {\em preorder} traversal of $T$) Bille et al. showed how to find the (preorder number of) node $\LCA(i,j)$ in $O(\log{n})$ time. Therefore, the only change required in order to adapt their data structure to our needs is the representation of nodes by their {\em inorder} rather than preorder numbers.   
   
The {\em local preorder number} $u_C$ of a node $u$ in $T$ and a cluster $C$ in $\mathcal{T}$ is the preorder number of $u$ in a preorder traversal of the cluster $C$. To find the preorder number of $\LCA(i, j)$ in $O(\log{n})$ time, Bille et al. showed that it suffices if for any node $u$ and any cluster $C$ we can compute $u_C$ in constant time from $u_A$ or $u_B$ (the local preorder numbers of $u$ in the clusters $A$ and $B$ whose merge is the cluster $C$) and vice versa. In Lemma 6 of~\cite{top-tree} they show that indeed they can compute this in constant time. The following lemma is a modification of that lemma to work when $u_A, u_B$ and $u_C$ are  {\em local inorder numbers}.

\begin{lemma}[Modified Lemma 6 of~\cite{top-tree}]
\label{lem:local_inorder}
Let $C$ be an internal node in $\mathcal{T}$ corresponding to the cluster obtained by merging clusters $A$ and $B$. For any node $u$ in $C$, given $u_C$ we can tell in constant time if $u$ is in $A$ (and obtain $u_A$) in $B$ (and obtain $u_B$) or in both. Similarly, if $u$ is in $A$ or in $B$ we can obtain $u_C$ in constant time from $u_A$ or $u_B$.
 
\end{lemma}
\begin{proof}
 We show how to obtain $u_A$ or $u_B$ when $u_C$ is given. Obtaining $u_C$ from $u_A$ or $u_B$ is done similarly. For each node $C$, we store a following information: 
\begin{itemize}
\item $\ell(A)$ ($r(A)$): the first (last) node visited in an inorder traversal of $C$ that is  also a node in $A$. 
\item $\ell(B)$ ($r(B)$): the first (last) node visited in an inorder traversal of $C$ that is  also a node in $B$. 
\item the number of nodes in $A$ and in $B$. 
\item $u'_C$, where $u'$ is the common boundary node of $A$ and $B$.  
\end{itemize}

\noindent Consider the case where $C$ is obtained by merging $A$ and $B$ vertically (when the bottom boundary node of $A$ is the top boundary node of $B$), and where $B$ includes vertices that are in the left subtree of this boundary node, the other case is handled similarly: 
\begin{itemize}
\item if $u_C < \ell(B)$ then $u$ is a node in $A$ and $u_A = u_C$. 
\item if $\ell(B) \le u_C \le r(B)$ then $u$ is a node in $B$ and $u_B = u_C-\ell(B)+1$. 
For the special case when $u_C = u'_C$ then $u$ is also the bottom boundary node in $A$ and $u_A = \ell(B)$.
\item if $u_c > r(B)$ then $u$ is a node in $A$ visited after visiting all the nodes in $B$ then $u_A = u_C -|B|+1$. 
\end{itemize}

\noindent When $C$ is obtained by merging $A$ and $B$ horizontally (when $A$ and $B$ share their top boundary node and $A$ is to the left of $B$):
\begin{itemize}
\item if $u_C  < r(A)$ then $u$ is a node in $A$ and $u_A = u_C$.
 
\item if $u_C \ge r(A) $ then $u$ is a node in $B$ and $u_B = u_C-|A|+1$. For the special case when $u_C = u'_C$ then $u$ is also the top boundary node in $A$ and $u_A = |A|$.\qed
\end{itemize}
\end{proof}

To complete the proof of Theorem~\ref{thm:rmq_ct}, we now explain how to use Lemma~\ref{lem:local_inorder}, given the inorder numbers of nodes $x$ and $y$ in $T$, to compute in $O(\texttt{depth}(\mathcal T))$ time the inorder number of $\LCA(x,y)$ in $T$. This is identical to the procedure of Bille et al. (except for replacing preorder with inorder) and is given here for completeness. 

We begin with a top-down search on $\mathcal T$ to find the first cluster whose top boundary node is $\LCA(x,y)$ (or alternatively to reach a leaf cluster whose top or bottom boundary node is $\LCA(x,y)$). At each cluster $C$ in the search we compute the local inorder numbers $x_C$ and $y_C$ of $x$ and $y$ in $C$. Initially, for the root cluster $T$ we set $x_T = x$ and $y_T = y$. If we reach a leaf cluster $C$ we stop the search. Otherwise, $C$ is an internal cluster with children $A$ and $B$. If $x_C$ and $y_C$ are in the same child cluster, we continue the search in that cluster after computing the new local inorder numbers of $x$ and $y$ in the appropriate child cluster (in constant time using Lemma~\ref{lem:local_inorder}). Otherwise, $x_C$ and $y_C$ are in different child clusters. If $C$ is a horizontal merge then we stop the search. If $C$ is a vertical merge (where $A$'s bottom boundary node is $B$'s top boundary node) then we continue the search in $A$ after setting the local inorder number of the node (either $x$ or $y$) that is in $B$ to be the bottom boundary node of $A$.

After finding the cluster $C$ whose top boundary node $v$ is $\LCA(x,y)$, we have the inorder number of $v$ in $C$ and we need to compute the inorder number of $v$ in the entire tree $T$. This is done by repeatedly applying Lemma~\ref{lem:local_inorder} on the path in $\mathcal T$ from $C$ to the root of $\mathcal T$.

\section{Compressing the String vs. the Cartesian Tree }
In this section we compare the sizes of the SLP compression $\mathcal{S}$ and the top-tree compression $\mathcal{T}$.

\subsection{An upper bound}\label{sec:compare upper}

We now show that given any SLP $\mathcal{S}$ of height $h$, we can construct a top-tree compression $\mathcal{T}$ based on $\mathcal{S}$ (i.e. non-greedily) such that $|\mathcal{T}| = O(|\mathcal{S}| \cdot \sigma)$ and the height of $\mathcal{T}$ is $O(h \log \sigma)$. Using $\mathcal{T}$, we can then answer range minimum queries on $S$ in time $O(h \log \sigma)$ as done in Section~\ref{sec:CT}. Furthermore, we can construct $\mathcal{T}$ in $O(n \log{\sigma} + |\mathcal{S}| \cdot \sigma)$ time using Rytter's SLP~\cite{Rytter03} as $\mathcal{S}$. Then, the height of $\mathcal{S}$ is $h=\log n$ and the size of $\mathcal{S}$ is larger than the optimal SLP by at most a multiplicative $\log n$ factor. 

Consider a rule $C\rightarrow AB$ in the SLP. We will construct a top-tree (a hierarchy of clusters) of $C$ (i.e. of the Cartesian tree $CT(C)$ of the string derived by the SLP variable $C$) assuming we have the top-trees of (the Cartesian trees of the strings derived by) $A$ and of $B$. 
We show that the top-tree of $C$ contains only $O(\sigma)$ new clusters that are not clusters in the top-trees of $A$ and of $B$, and that the height of the top-tree is only $O(\log \sigma)$ larger than the height of the top-tree of $A$ or the top-tree of $B$.
To achieve this, for any variable $A$ of the SLP, we will make sure that certain clusters (associated with its rightmost and leftmost paths) must be present in its top-tree. See Figure~\ref{figure:topdag}.

\begin{figure}[htp]
\centering
\hspace*{-0.3cm}
\ifarxiv \includegraphics[scale=0.9]{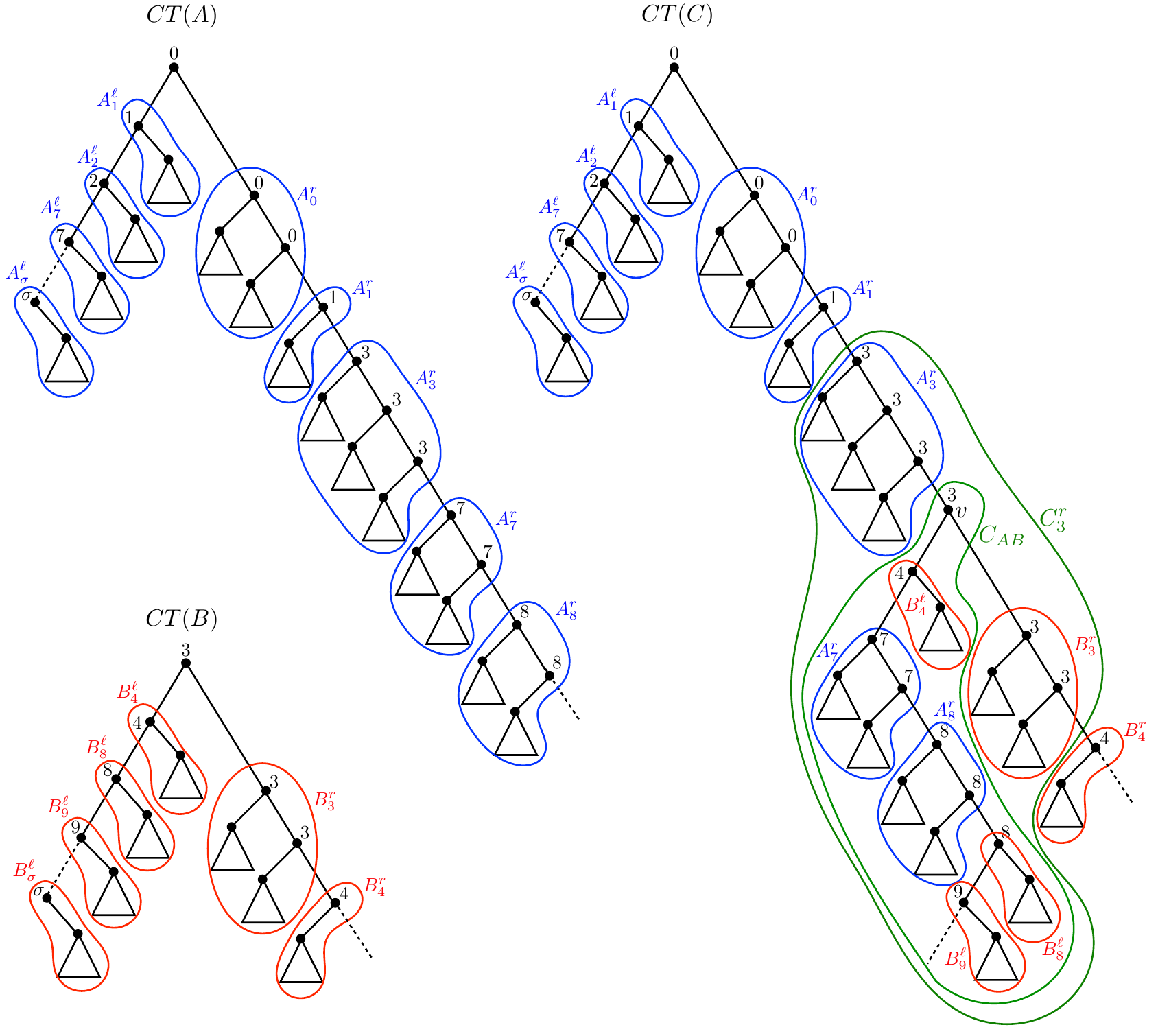}
\else \includegraphics[scale=0.74]{imageC.pdf}
\fi

\caption{The Cartesian tree of SLP variables $A,B,C$  where $C\rightarrow AB$. The 
single additional cluster $C_3^r$ (in green) is formed by merging existing clusters from $A$ (in blue) and from $B$ (in red). First, cluster $C_{AB}$ (corresponding to the tree $CT(A_sB_p)$) is formed by alternating subpaths of the leftmost path in $CT(B)$ and the rightmost path in $CT(A)$. 
Then, $C_{AB}$ is merged with $B_3^r$, $v$, and $A_3^r$. In this example, $A_s = \{A_i^r\ |\  i>3 \}$ and $B_p = \{B_i^\ell\ |\  i>3 \}$.
\label{figure:topdag}}
\end{figure}

\paragraph{\bf The structure of \boldmath$CT(C)$} We first describe how the Cartesian tree $CT(C)$ of the string derived by variable $C$ can be described in terms of the Cartesian trees $CT(A)$ and $CT(B)$. 
We label each node in a Cartesian tree with its corresponding character in the string. These labels are only used for the sake of this description, the actual Cartesian tree is an unlabeled tree. By definition of the Cartesian tree, the labels are monotonically non-decreasing as we traverse any root-to-leaf path. 
Let $\ell(A)$ (respectively $r(A)$) denote the path in $CT(A)$ starting from the root and following left (respectively right) edges. 
Since we break ties by taking the leftmost occurrence of the same character we have that the path $\ell(A)$ is strictly increasing (the path $r(A)$ is just non-decreasing).

Let $x$ be the label of the root of $CT(B)$. To simplify the presentation we assume that the label of the root of $CT(A)$ is smaller or equal to $x$ (the other case is symmetric). 
 Split $CT(A)$ by deleting the edge connecting the last node on $r(A)$ that is smaller or equal to $x$ with its right child. 
 The resulting two subtrees are the Cartesian trees $CT(A_p)$ and $CT(A_s)$ of a prefix $A_p$ and a suffix  $A_s$ of $A$ whose concatenation is $A$. The prefix $A_p$ ends at the last character of $A$ that is at most $x$. Split $CT(B)$ by deleting the edge connecting the root to its right child. The resulting two subtrees are the Cartesian trees $CT(B_p)$ and $CT(B_s)$ of a prefix and a suffix of $B$. The prefix $B_p$ ends with the first occurrence of $x$ in $B$. 
 
 The Cartesian tree $CT(C)$ of the concatenation $C = AB$ can be described as follows: Consider the Cartesian trees $CT(A_p)$ of $A_p$, $CT(B_s)$ of $B_s$, and $CT(A_sB_p)$ of the concatenation of $A_s$ and $B_p$. 
It is easy to verify that the root of $CT(A_sB_p)$ has no right child. 
Attach $CT(B_s)$ as the right child of the root of $CT(A_sB_p)$. Then attach the resulting tree as the right child of the last node of the rightmost path of $CT(A_p)$. See Figure~\ref{figure:topdag}. 

The above structural description of $CT(C)$ is not enough for our purposes. In particular, a recursive computation of $CT(A_sB_p)$ would lead to a linear $O(\sigma)$ increase in the height of the top-tree of $CT(C)$ compared to that of $CT(A)$ and $CT(B)$. In order to guarantee a logarithmic $O(\log \sigma)$ increase, we need to describe the structure in more detail.

For a node $v$ with label $i$ appearing in  $\ell(A)$ other than the root of $A$, we define $A_i^\ell$ subtree rooted at the $v$'s right child, together with $v$. Next consider the path $r(A)$. For every label $i$ there can be multiple vertices with label $i$ that are consecutive on $r(A)$. We define $A_i^r$ to be the subtree of $A$ induced by the union of all vertices of $r(A)$ that have label $i$ together with the subtrees rooted at their left children. Again, we treat the first node of $r(A)$ (i.e. the root of $CT(A)$) differently: if the label of the root is $i$ then $A_i^r$ does not include the root nor its left subtree. See Figure~\ref{figure:topdag} (left).

It is easy to see that $CT(A_p)$ consists of the subtree of $A$ induced by all the $A_i^\ell$'s and all the $A_j^r$ for $j \leq x$. See Figure~\ref{figure:topdag} (right). It is also easy to see that $CT(B_s)$ consists of all the $B_i^r$'s. The structure of $CT(A_sB_p)$ is a bit more involved. It consists of alternations of $A_i^r$'s and $B_i^\ell$'s which we describe next.

We describe the structure of $CT(A_sB_p)$ constructively from top to bottom. This constructive procedure is just for the sake of describing the structure of $CT(A_sB_p)$. We will later describe a different procedure for constructing the clusters of the corresponding top-trees. 
The root of $CT(A_sB_p)$ is the root of $B$. Initially, the root is marked L (indicating that the root can only obtain a left child). Throughout the procedure we will make sure that exactly one node is marked (by either L or R). 
For increasing values of $i$, starting with $i=x+1$ and ending when $i$ exceeds $\sigma$, we do the following: 
\begin{enumerate}
	\item if $A_i^r$ is defined: 
	\begin{enumerate}
		\item attach $A_i^r$ as the left child of the marked node if the marked node is marked with L, and as the right child otherwise, 
		\item unmark the marked node, and instead mark the last node on the rightmost path of $A_i^r$ with R. 
	\end{enumerate}
	\item if $B_i^\ell$ is defined: 
	\begin{enumerate}
		\item attach $B_i^\ell$ as the left child of the marked node if the marked node is marked with L, and as the right child otherwise, 
		\item unmark the marked node, and instead mark the root of $B_i^\ell$ with L. 
	\end{enumerate}
\end{enumerate}

Note that $CT(A_sB_p)$ consists of all subtrees $B_i^\ell$, the subtrees $A_i^r$ for $i>x$, and $O(\sigma)$ additional edges (these are the edges that were created when attaching subtrees to a marked node during the construction procedure). Also observe that each subtree $A_i^r$ is incident to at most two new edges, one incident to the root of $A_i^r$ and the other incident to the rightmost node of $A_i^r$. Similarly, each subtree $B_i^\ell$ is incident to at most two new edges, both incident to the root of $B_i^\ell$. Imagine contracting each $A_i^r$ and each $B_i^\ell$ into a single node. Then the result would be a single ``zigzag'' path of new edges. We will next use these properties when describing the clusters of $CT(C)$.

\paragraph{\bf The clusters of \boldmath$CT(C)$}


We next describe how to obtain the clusters for the top-tree of $CT(C)$ from the the clusters of the top-trees of $CT(A)$ and $CT(B)$. 
For each variable (say $A$) of the SLP $\mathcal S$ of $S$, we require that in the top-tree of $S$ there is a cluster for every $A_i^\ell$ and every $A_i^r$. 
Clusters for $A_i^\ell$ only have a top boundary node.
Clusters for $A_i^r$ have a top boundary node (the root of $A_i^r$), and a bottom boundary node (the last node on the rightmost path of $A_i^r$). 
We will show how to construct all the $C_i^\ell$ and $C_i^r$ clusters of $C$ by merging clusters of the form $A_i^\ell, A_i^r,B_i^\ell$, and $B_i^r$ while introducing only $O(\sigma)$ new clusters, and with $O(\log \sigma)$ increase in height. 
First observe that, by the structure of $CT(A_p)$, we have that, for every $i$, $C_i^\ell=A_i^\ell$, so we already have these clusters. Next consider the clusters $C_i^r$. Recall that $x$ denotes the label of the root of $CT(B)$. By the structure of $CT(A_p)$, $C_i^r = A_i^r$ for every $i<x$ and, by the structure of $CT(B_s)$, $C_i^r = B_i^r$ for every $i>x$. Therefore, the only new cluster we need to create is $C_x^r$. 

The cluster $C_x^r$ corresponds to a subtree of $CT(C)$ that is composed of the following components: First, it contains the cluster $A_x^r$. Then, $CT(A_sB_p)$ is connected as the right child of the rightmost node of $A_x^r$. Finally, $B_x^r$ is connected as the right child of the root of $CT(A_sB_p)$. Since we already have the clusters for $A_x^r$ and $B_x^r$, we only need to describe how to construct a cluster $C_{AB}$ corresponding to $CT(A_sB_p)$. The cluster $C_x^r$ will be obtained by merging these three clusters together.

Recall from the structural discussion of $CT(A_sB_p)$ that  
$CT(A_sB_p)$ consists of all subtrees $B_i^\ell$ and the subtrees $A_i^r$ for $i>x$. We already have the clusters for these $O(\sigma)$ subtrees. These subtrees are connected together to form $CT(A_sB_p)$ by $O(\sigma)$ new edges that are only incident to the boundary nodes of the clusters. We will merge the existing clusters to form the cluster $C_{AB}$ by creating $O(\sigma)$ new clusters, but only increasing the height of the top-tree by $O(\log \sigma)$.
Performing the merges linearly would increase the height of the top-tree by $O(\sigma)$. Instead, the merging process consists of $O(\log \sigma)$ phases. In each phase we choose a maximal set of disjoint pairs of clusters that need to be merged and perform these merges. Since $CT(A_sB_p)$ is a binary tree, we can perform half the remaining merges in each phase. This process can be described by a binary tree of height $O(\log \sigma)$ whose leaves are the $O(\sigma)$ clusters $B_i^\ell$ and $A_i^r$ that we started with. 

To conclude, given the clusters $A_i^\ell, A_i^r, B_i^\ell, B_i^r$, we have shown how to compute all clusters $C_i^\ell, C_i^r$. Once we have all clusters of the SLP's start variable, we merge them into a single cluster (i.e. obtain the top-tree of the entire Cartesian tree of $S$) by merging all its $O(\sigma)$ clusters (introducing $O(\sigma)$ new clusters and increasing the height by $O(\log \sigma)$) similarly to the description above. This concludes the proof of Theorem~\ref{thm:tg}.

\subsection{A lower bound}\label{sec:compare lower}

We now prove Theorem~\ref{thm:lower bound}. That is, for every sufficiently large $\sigma$ and $s=\Omega(\sigma^{2})$
we will construct a string $S$ of integers in $[0,\sigma)$ that
can be described with an SLP $\mathcal{S}$ of size $s$, such that any top-tree compression $\mathcal{T}$ of the Cartesian
tree $\mathcal{C}$ of $S$ is of size $\Omega(s\cdot \sigma)$.

Let us first describe the high-level intuition. The shuffle of two strings $x[1..\ell]$ and $y[1..\ell]$ is defined as
$x[1]y[1]x[2]y[2]\ldots x[\ell]y[\ell]$. It is not very difficult to construct a small SLP describing a collection of many
strings $A_{i}$ and $B_{j}$ of length $\ell$,
and choose $s$ pairs $(i_{k},j_{k})$ such that every SLP describing all shuffles of
$A_{i_{k}}$ and $B_{j_{k}}$ contains $\Omega(s\cdot \ell)$ nonterminals. However, our goal is to show a lower bound
on the size of a top-tree compression of the Cartesian tree, not on the size of an SLP. This requires designing the
strings $A_{i}$ and $B_{j}$ so that a top-tree compression of the Cartesian tree of $A_{i}B_{j}$ roughly corresponds
to an SLP describing the shuffle of $A_{i}$ and $B_{j}$.

Let $\sigma'=\lfloor (\sigma-4)/2 \rfloor $ and $\ell$ be a parameter such that $2^{\ell}-1 \geq \sigma'$.
We start with constructing $2^{\ell}-1$ distinct auxiliary strings $X_{i}$ over $\{\sigma-2,\sigma-1\}$, each of the
same length $\ell$. We construct every such string except for $(\sigma-1)^{\ell}$, so that Cartesian trees corresponding
to $X_{i}$s are all distinct. The total number of $X_{i}$s is $2^{\ell}-1$ and there is an SLP $\mathcal{X}$ of size $O(2^{\ell})$ that
contains a nonterminal deriving every $X_{i}$.
Next, let $X_{i,j}$ denote the string $X_{i} (\sigma-3) X_{j}$. By construction, Cartesian trees corresponding to $X_{i,j}$s
are all distinct, and all $X_{i,j}$s are of the same length.

The second and the third step are symmetric. We construct strings $A_{i}$ of the form:
\[ X_{1,i} 2 X_{2,i} 4 \cdots (2\sigma'-4) X_{\sigma'-1,i} (2\sigma'-2) X_{\sigma',i} (2\sigma') \]
for every $i=1,2,\ldots,2^{\ell}-1$.
There are $2^{\ell}-1$ such strings $A_{i}$, and there is an SLP $\mathcal{A}$ of size $O((2^{\ell}-1) \cdot \sigma')$
that contains a nonterminal deriving every $A_{i}$.

Similarly, we construct strings $B_{j}$ of the form:
\[ (2\sigma'-1) X_{\sigma',j} (2\sigma'-3) \ldots 3 X_{2,j} 1 X_{1,j} .\]

Finally, we obtain $S$ by concatenating the strings $A_{i} B_{j} 0$ for all $i,j=1,2,\ldots,2^{\ell}-1$.
The total size of an SLP that
generates $S$ is $O(2^{\ell}\cdot \sigma'+4^{\ell})$. It remains to analyze the size of a top-tree compression $\mathcal{T}$ of
the Cartesian tree $\mathcal{C}$ of $S$.

We first need to understand the structure of $\mathcal{C}$. Because all strings $A_{i}B_{j}$ are separated by $0$s, the Cartesian tree of
$S$ consists of a right path of length $(2^{\ell}-1)^{2}$ and the Cartesian tree of $A_{i} B_{j}$ attached as the left subtree of the
$((i-1)(2^{\ell}-1)+j)$-th node of the path.
The Cartesian tree of a string $A_{i} B_{j}$ consists of a path of length $2\sigma'$ starting at the root
and consisting of nodes $u_{1} - v_{1} - u_{2} - v_{2} - \ldots - u_{\sigma'} - v_{\sigma'}$ such that $v_{i}$ is
the left child of $u_{i}$ and $u_{i+1}$ is the right child of $v_{i}$. For every $k\in\{1,\ldots,\sigma'\}$,
the right subtree of $u_{k}$ is the Cartesian tree of $X_{k,i}$
and the left subtree of $v_{k}$ is the Cartesian tree of $X_{k,j}$.
See Figure~\ref{fig:bc}.

\begin{figure}[htb]
\centering
\ifarxiv \includegraphics[width=0.6\textwidth]{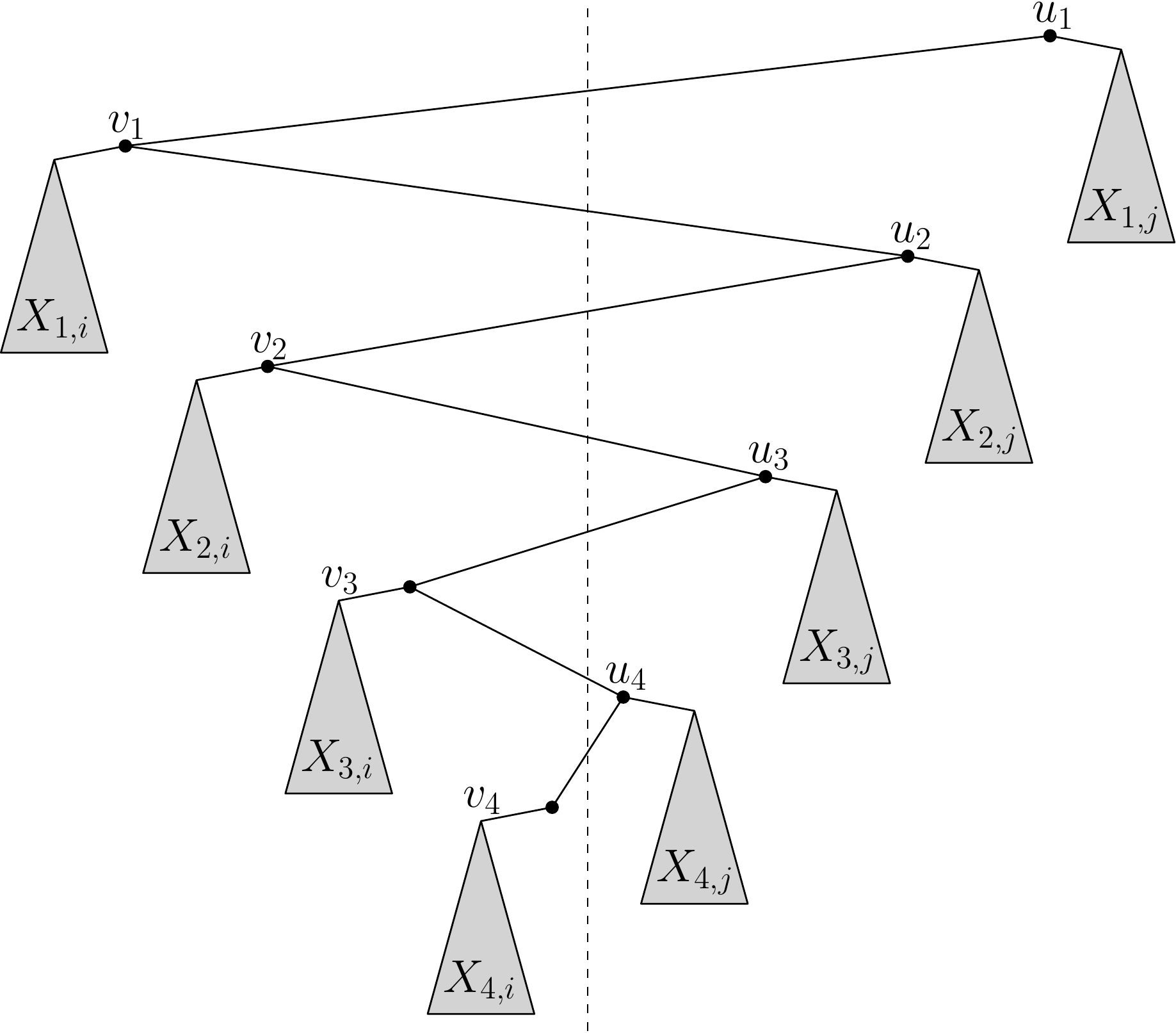}
\else \includegraphics[width=0.7\textwidth]{bc}
\fi
\caption{Structure of the Cartesian tree of $A_{i}B_{j}$ for $\sigma'=4$.}
\label{fig:bc}
\end{figure}

We define a {\it zigzag} to be an edge $u-v$
such that $v$ is the left child of $u$. Furthermore, for some
$k\in\{1,\ldots,\sigma'\}$ and $i,j\in\{1,\ldots,2^{\ell}-1\}$,
the right subtree of $u$ should be the
Cartesian tree of $X_{k,j}$, while the left subtree of $v$ should be the Cartesian tree
of $X_{k,i}$.

\begin{proposition}
\label{prop:many}
The Cartesian tree of $A_{i}B_{j}$ contains $\sigma'$ distinct zigzags. Furthermore, a zigzag contained in the Cartesian tree of $A_{i}B_{j}$ is not contained in the Cartesian tree of $A_{i'}B_{j'}$ for any $i'\neq i$ or $j'\neq j$.
\end{proposition}

\begin{lemma}
\label{lem:generate}
If $\mathcal{T}$ is a top-tree compression of a tree $T$ with $x$ distinct zigzags then $|\mathcal{T}|  = \Omega(x)$.
\end{lemma}

\begin{proof}
We associate each distinct zigzag with a smallest cluster of the top-tree of $T$ that contains it. We claim that each
cluster obtained by merging clusters $A$ and $B$ is associated with $O(1)$ zigzags.
Since the size of $\mathcal T$ equals the number of distinct clusters in the top-tree of $T$, the lemma follows.
Consider a zigzag $z = u - v$ associated with $A\cup B$. Hence, $z$ is not contained in $A$ nor in $B$. We consider two cases.
\begin{enumerate}
\item $A$ and $B$ are merged horizontally. Then $A$ and $B$ share the top boundary node $b$, and in fact $u=b$.
It follows that $z$ is the only zigzag in $A\cup B$ that is not in $A$ nor in $B$.
\item $A$ and $B$ are merged vertically. Then the top boundary node of $A$ is the bottom boundary node $b$ of $B$.
Then either $b=u$, $b=v$, or $b$ is a node of the Cartesian tree of some
$X_{k,x}$ attached as the right subtree of $u$ or the left subtree of $v$.
Each of the first two possibilities gives us one zigzag associated with $A\cup B$ that is not in $A$ nor in $B$.
In the remaining two possibilities (i.e. when the Cartesian tree of $X_{k,x}$ is attached as the right subtree of $u$ or as the left subtree of $v$), because the size of the Cartesian tree of every $X_{k,x}$
is the same, we can determine $u$ or $v$, respectively, by navigating up from $b$ as long as the size of
the current subtree is too small, and proceed as in the previous two cases.\qed
\end{enumerate}
\end{proof}

Combining Proposition~\ref{prop:many} and Lemma~\ref{lem:generate} we conclude that $\mathcal{T} = \Omega(4^{\ell} \cdot \sigma')$.
Recall that the size of an SLP that generates $S$ is $O(2^{\ell}\cdot \sigma'+4^{\ell})$, where
$\sigma'=\lfloor (\sigma-4)/2\rfloor = \Theta(\sigma)$ and $\ell$ is parameter such that $2^{\ell}-1 \geq \sigma'$.
Given a sufficiently large $\sigma$ and $s=\Omega(\sigma^{2})$, we first choose $\ell = \lceil 1/2\log s \rceil$.
Observe that then $2^{\ell}-1 \geq \sigma'$ indeed holds because of the assumption $s=\Omega(\sigma^{2})$.
We construct a string $S$ generated by an SLP of size $O(2^{\ell}\cdot \sigma'+4^{\ell})=O(s)$,
and any top-tree compression $\mathcal{T}$ of the Cartesian tree of $S$ has size $\Omega(s\cdot \sigma')$.
This concludes the proof of Theorem~\ref{thm:lower bound}.

\ifarxiv 

\else \bibliography{ref}
\fi

\end{document}